\documentclass[letterpaper]{article} 
\usepackage[preprint]{aaai2027}  
\usepackage[hyphens]{url}  
\usepackage{graphicx} 
\usepackage{natbib}  
\usepackage{caption} 
\usepackage{booktabs}
\usepackage{array}
\usepackage{amsmath,amssymb,amsthm}
\usepackage{tikz}
\usetikzlibrary{fit,positioning}
\newtheorem{proposition}{Proposition}

\title{The Working Set of a Coding Agent:\\Coherence Debt in Repository-Scale Tasks}

\author{
    Bardia Mohammadi\textsuperscript{\rm 1},
    Lars Klein\textsuperscript{\rm 2},
    Aman Chadha\textsuperscript{\rm 3},
    Akhil Arora\textsuperscript{\rm 4},
    Laurent Bindschaedler\textsuperscript{\rm 1}
}
\affiliations{
    \textsuperscript{\rm 1}Max Planck Institute for Software Systems
    \textsuperscript{\rm 2}EPFL
    \textsuperscript{\rm 3}Apple\thanks{Work done outside of role at Apple.}
    \textsuperscript{\rm 4}Aarhus University\\
    \{bmohammadi, bindsch\}@mpi-sws.org
}

\begin{document}
\maketitle

\begin{abstract}
Repository-scale coding requires an agent to keep tests, imports, configuration, and migration rules consistent within a bounded context window. We model this as reconstructing a \emph{coupled-fact graph}: at each edit, a required fact comes from recent context or parametric memory, and the facts covered by neither form \emph{coherence debt}. We supply and withhold each channel and inject faults across seven models and five harnesses. As expected, no model completes a task on an unseen API with both channels empty, and putting the facts in the prompt restores success. When a rename defeats what models memorized about a real library, all seven fail in the same place, passing and missing the same tests. Availability decides the outcome and distance does not: withholding a fact costs exactly the work it supports, and a supplied fact works as well far from the edit as next to it. Harnesses pay unequal prices for it: configurations that all pass every test differ more than tenfold in tokens consumed because they rebuild the same content at different rates, and spending more recovers nothing when facts are withheld. A missing fact produces wrong work rather than absent work: an agent asked to act acts, fabricating the file or guessing the value, so instruments built on reads look for a hole already filled. How often it says it is blocked instead is a property of the model, from every trial to none. Availability does not settle every edit: where standard and code disagree, agents follow the standard even when it prescribes the worse code, so a stale convention file costs more than no file. Because parametric memory substitutes for reading, on SWE-bench, where models likely know the repositories, reads no longer predict success. Harnesses should keep the facts an edit depends on available when the agent writes, and check that availability against what the agent produces rather than what it reads.
\end{abstract}

\section{Introduction}
\label{sec:intro}

Coding agents increasingly attempt migrations, upgrades, and bug fixes that span a repository. Such work is difficult for a reason that is easy to miss when evaluation observes only the final patch: an edit in one file is correct only if it remains consistent with facts elsewhere. A validator must agree with its tests, a renamed symbol with every import, and a configuration change with the runtime contract. No one hands the agent this structure. The agent must reconstruct it through reads, edits, tests, and sometimes handoffs to other workers.

We study the availability of those coupled facts \emph{at the moment of an edit}. A tool-using agent has two ways to obtain a required fact. The fact may sit in recent context because the agent read it or the harness supplied it, or the model may already hold it in parametric memory. The first channel is bounded and evicts facts as new reads arrive, while the second is fixed for a model and unreliable on novel application programming interfaces (APIs). If neither channel covers a fact, the edit proceeds from an incomplete view. We call this shortfall \emph{coherence debt}.

This framing adapts the working-set idea of virtual memory~\citep{Denning1968WorkingSet} to an addressless setting. Repository facts behave unlike pages: a migration rule stated in prose and a stale implementation may encode different versions of one fact, yet no address or invalidation bit tells the agent which is current. More context therefore helps only when it contains the facts coupled to the present edit. Likewise, decomposition helps only when the partition leaves mutually consistent facts together.

Existing work establishes repository-level evaluation~\citep{Jimenez2024SWEBench}, builds tool-using loops~\citep{Yang2024SWEAgent,Wang2025OpenHands}, guides retrieval with repository structure~\citep{Ouyang2025RepoGraph}, and shows that successful agents gather context before editing~\citep{Mehtiyev2026BehavioralDrivers}. We ask a narrower causal question: \emph{which facts must be available when the agent writes, and can recent context and model prior substitute for one another?}

We manipulate both channels, then measure how they operate in trajectories. Four fictional API migrations, each with 12 mechanically checked requirements, run closed-book: the model receives the task description but no workspace and no tools, so both channels start empty. A matched front-loaded condition puts the exact rules and source files in the prompt, with the model and tool surface unchanged. A real Pydantic migration with 79 tests, and a twin with every API name renamed, hold the task fixed while making memorized knowledge useless. In tool-using runs we measure which of an edited file's dependencies the agent read shortly before the edit, and we build synthetic tasks whose required facts we know because we author them. The Experimental Design section gives the full setup.

We draw five claims from this evidence. First, context substitutes for missing knowledge: across 154 closed-book trials on the four fictional migrations no model completes one, and putting the same facts in the prompt lifts 299 of 300 matched trials to at least 9 of the 12 requirements. Second, models that memorized the same library fail in the same place: on the renamed migration, 66 of 70 trials across seven models end at the same score, each passing the identical 24 of the 79 tests. Third, availability decides the outcome and distance does not: withholding required facts costs exactly the work they support and no more, so damage falls linearly across 30 trials, while a supplied fact is used as reliably at the far end of a $128{,}000$-character context as beside the edit. Fourth, small exploratory runs indicate that these effects act through the coupling of the task: irrelevant content hurts only when it sits inside the files the agent must read, and splitting the work across subagents hurts a tightly coupled task while leaving independent fixes intact. Fifth, harnesses pay unequal prices for the same result: across configurations that all pass every test, the tokens consumed over a run differ by $12.8\times$ while the amount held in front of the model at any moment differs by only $1.8\times$, and the expensive configurations recover nothing extra when facts are withheld. Throughout, a missing fact produces wrong work instead of absent work: the agent fabricates the file or guesses the value, and how often it reports being blocked instead ranges from every trial to none across the models we test. Two scope conditions matter just as much. Availability does not decide the edit on its own: when a written standard and working code state the same fact differently, agents follow the standard in every one of 39 trials across two harnesses, even where that means writing the worse code. And because parametric memory substitutes for reading, on SWE-bench, where the models likely already know the repositories and we observe the task structure only coarsely, what an agent reads no longer predicts whether it succeeds.

\paragraph{Contributions.}
We (1) show that withholding facts costs exactly the work they support, and derive from a two-channel, edit-time account of repository coherence why read-derived instruments cannot see that shortfall, namely that they must overstate missing facts by exactly the parametric coverage, and we confirm it by enumeration; (2) isolate channel substitution with controlled closed-book, rename, and front-load interventions; (3) establish that coverage acts through presence, by fault injection on tasks whose coupled facts we construct; (4) bound the account from within our own workloads, showing coverage does not determine the edit when covered facts conflict, since agents follow a written standard over contradicting code even where it is the worse guide; (5) report that the framework does not transfer to real repositories rather than treat a within-workload diagnostic as universal; and (6) show that read-derived instruments cannot measure the shortfall they were built for, since agents compensate for a missing fact by acting, and give the consequences for harness design, including that a stale convention file costs more than no file at all. We provide the workload rules and trial ledger (Appendix~\ref{app:materials}), notation (Appendix~\ref{app:notation}), instrumentation (Appendix~\ref{app:measurement}), reproduction details (Appendix~\ref{app:reproducibility}), and the ethics and data statement (Appendix~\ref{app:ethics}) in the Appendix.

\section{Coherence as Edit-Time Coverage}
\label{sec:model}

\paragraph{Coupled-fact graph.}
A task $T$ induces a graph $G_T=(V_T,E_T)$. The nodes are atomic facts relevant to the patch: symbols, tests, configuration values, imports, migration rules, and invariants. An edge means that relying on or changing one endpoint requires the other to remain consistent. Let $C_T\subseteq V_T$ be a minimal set that must be jointly correct for the task oracle to pass, and let $C_T^{(i)}$ be the subset required by edit $e_i$. $C_T$ is relative to the oracle, the fact schema, and the implementation path, so alternative correct patches can induce different minimal fact sets. Tests and prose rules couple files that share no static dependency, so $G_T$ extends well beyond the import graph.

At edit time $t_i$, $R_{t_i}$ denotes facts resident in the effective context and $K_M$ facts available from model $M$'s parametric memory. We define per-edit coherence debt as
\begin{equation}
D(e_i)=\left|C_T^{(i)}\setminus\left(R_{t_i}\cup K_M\right)\right|.
\label{eq:event-debt}
\end{equation}
The union carries the key prediction: a fact may arrive through either channel, so success should depend on coverage and not on which channel supplied it. Uncovered coupled facts create a specific, measurable risk. There are other causes of failure, and coverage alone does not guarantee a correct edit.

\begin{figure}[t]
\centering
\begin{tikzpicture}[font=\scriptsize,>=stealth,
  fact/.style={circle,draw=black!70,minimum size=5.5pt,inner sep=0pt},
  req/.style={fact,fill=black!60},
  box/.style={draw,rounded corners=2pt,inner sep=3pt,align=center}]
\node[req] (f1) at (0.2,1.15) {};
\node[req] (f2) at (0.9,1.42) {};
\node[req] (f3) at (0.65,0.72) {};
\node[fact] (f4) at (1.45,1.05) {};
\node[fact] (f5) at (1.30,0.35) {};
\draw[black!50] (f1)--(f2) (f2)--(f3) (f1)--(f3) (f3)--(f5) (f2)--(f4);
\node[draw,rounded corners=3pt,fit=(f1)(f2)(f3)(f4)(f5),inner sep=4pt] (G) {};
\node[anchor=south] at (G.north) {$G_T$};
\node[box] (edit) at (3.35,0.9) {edit $e_i$};
\draw[->] (G.east)--node[above] {$C_T^{(i)}$}(edit.west);
\node[box] (R) at (5.45,1.45) {recent context\\$R_{t_i}$};
\node[box] (K) at (5.45,0.35) {model prior\\$K_M$};
\draw[->] (R.west)--(edit.north east);
\draw[->] (K.west)--(edit.south east);
\node at (3.35,-0.05) {$D(e_i)=|C_T^{(i)}\setminus(R_{t_i}\cup K_M)|$};
\draw[->,dashed] (edit.south)--(3.35,0.18);
\end{tikzpicture}
\caption{Each edit requires a slice of the task's coupled-fact graph. Recent context and model prior are substitutable channels, and their uncovered remainder is coherence debt.}
\label{fig:model}
\end{figure}
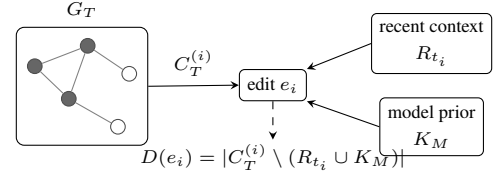

\paragraph{Trajectory and thrashing.}
Because $R_t$ is bounded, required facts enter and leave the effective working set. We call a trajectory \emph{thrashing} when uncovered edit-time debt survives repeated read--edit--test cycles because no new evidence retires it. A long successful trajectory may accumulate more raw debt than a short failure, so the relevant diagnostics are debt per edit and its mean over the final quarter of a trajectory, rather than unnormalized cumulative debt.

\paragraph{Observable proxy.}
Neither $G_T$ nor internal context retention is directly observable in historical harness logs. For an edit of file $f_i$, we build a static import graph and compute
\[
\rho_w(f_i,t_i)=
\frac{|N(f_i)\cap\mathrm{Read}(t_i-w,t_i)|}{|N(f_i)|},
\]
where $N(f_i)$ contains one-hop import neighbors and $\mathrm{Read}(t_i-w,t_i)$ contains files read during the preceding $w$ tool events. Reads after the edit never count, and we set $\rho_w=1$ when $|N(f_i)|=0$. This convention treats import-isolated files as covered by construction and cannot detect their non-import dependencies. We call $1-\rho_w$ the \emph{residency score} and reserve \emph{coherence debt} for the latent fact-level quantity of Equation~\ref{eq:event-debt}: the score counts unread neighbor files, not missing facts, and misses dynamic, prose, prompt-supplied, and parametric ones. We sweep $w\in\{4,8,16,32,64,128\}$ rather than reporting a single favorable window.

\paragraph{Predictions.}
We pre-specify four tests. (P1) If $R_t$ and $K_M$ are empty, novel migrations should reach a floor, and supplying the missing facts through either channel should lift it. (P2) Models sharing a partial prior should fail on the same task region, and not merely at the same rate. (P3) Coverage should act through presence rather than proximity: withholding required facts should cost only the work those facts support, while a fact's distance from the edit should not matter, and total read volume should not substitute for either. (P4) Placement and decomposition should matter through the coupled-fact graph: unavoidable in-file bloat should displace required facts, and partitioning across a coupled cut should add incoherence.

\subsection{A versioned event semantics}
\label{sec:events}

Historical tool logs expose actions and observations but not the facts that justified a write. To make the framework prospectively measurable, we define a run as a totally ordered stream of five event types. Each event carries a run identifier, monotone event identifier, timestamp, actor, and type-specific payload, and file versions advance on every edit, including a restoring edit. Table~\ref{tab:events} lists the payload each type must record. The actor emits an edit intent before the write, naming the required facts and versions it relies on. This records the actor's claimed support set, which exposes missing and stale support before the outcome but requires external validation rather than establishing ground truth.

\begin{table*}[t]
\centering
\scriptsize
\begin{tabular}{@{}p{0.15\linewidth}p{0.35\linewidth}p{0.42\linewidth}@{}}
\toprule
Event & Required payload & Role in coverage analysis \\
\midrule
\texttt{file\_read} & path, observed repository/file version, byte range & What source reached an actor; asserts no semantic extraction by itself. \\
\texttt{fact\_extracted} & fact identifier/value, source span and version, acquisition mode & Adds a versioned fact to the actor ledger by reading, prompt supply, or handoff. \\
\texttt{edit\_intent} & target/base version, required facts, relied-on fact versions, proposed transformation & Records the actor's claimed pre-write support set, which exposes missing and stale dependencies before the outcome but requires external validation. \\
\texttt{test\_feedback} & command, oracle result, diagnostics, contradicted intents & Ties environmental evidence to the edits it validates or invalidates. \\
\texttt{revert} & restored content/version, reverted intent, triggering feedback & An explicit compensating write, not a reversal inferred from a later diff. \\
\bottomrule
\end{tabular}
\caption{Prospective five-event representation. Historical experiments use coarser proxies. The schema specifies what a coherence-aware harness should retain.}
\label{tab:events}
\end{table*}

The stream lets us separate four causes that otherwise end in the same failed test: (i) a \emph{working-set miss}, where a required fact has no current extraction for the editing actor; (ii) a \emph{stale read}, which relies on a fact from an obsolete source version; (iii) a \emph{handoff gap}, where a current fact held by a parent never reaches the assigned worker; and (iv) a \emph{speculative write}, an uncovered intent later contradicted and undone. We measure parametric shortfall separately through closed-book behavior, since no harness logs internal $K_M$.

These distinctions matter: a reread repairs a missing or stale fact but not a handoff policy, serialization removes cross-worker gaps at the cost of time, and a stronger prior covers a fact with no read event. We formalize that last asymmetry as a measurement proposition, stated and proved in the Appendix: an event-only estimator observes read-covered facts but not parametric coverage, so it reports the union of the true uncovered set and the parametrically covered set, and its overstatement equals exactly the parametric coverage. We verify the implementation against that prediction by enumerating every read/parametric/uncovered assignment of required facts across 1,089 simulator runs (Appendices~\ref{app:estimator} and~\ref{app:event-simulator}). Of 124 fully covered runs, 119 pass the oracle while the estimator still reports working-set misses. A union-aware estimator is exact. So $\rho_w$ is a lower bound on observable read-derived coverage rather than an estimate of total coverage, and we keep the empirical score deliberately narrow: historical logs recover reads, edits, tests, and worker identity, but never structured fact extraction, so every headline trajectory result uses the residency score. We exclude a tempting ``retirement overrun'' term because its reliable estimator depends on the terminal outcome and would leak the label.

\section{Experimental Design}
\label{sec:design}

\paragraph{Channel-control workloads.}
Four hand-authored migrations, Sprocket (Rust), Grimwire (Go), Kestrix (Python), and Zynet (JavaScript), define fictional v1$\to$v2 APIs. Each contains three editable files, a change log, and 12 mechanically checked migration requirements. We invented their library names and rules, so prior exposure is highly unlikely, though generic idioms remain available. A standard closed-book prompt describes the requested change but withholds the workspace and tools, so $R_t\cap C_T\approx\emptyset$: the task description remains while its coupled facts are absent. In the matched front-loaded condition, the prompt also carries the verbatim change log and editable v1 sources. We leave the model and sandbox unchanged.

To that we add a 79-test Pydantic v1$\to$v2 migration, which $K_M$ plausibly covers, and an adversarial rename that preserves behavior while replacing Pydantic surface forms with a lexical shim. Seven model families participate in closed-book trials (Sonnet, Haiku, Codex, Z.ai, DeepSeek, Qwen, Gemini), and five participate in the fully sandboxed front-load matrix. Codex is absent from that matrix because of a startup-time sandbox interaction that precedes any model call, and Gemini because its quota expired before the matrix ran. Under a relaxed-sandbox runner Codex reaches $12/12$ on each of the four workloads, in 12 of 24 trials. We report that number in the Appendix as an instrumentation diagnostic, not as headline evidence.

\paragraph{Tool-using workloads.}
The primary tool-using corpus contains 122 matched trials across Claude Code, Codex CLI, Aider~\citep{Gauthier2023Aider}, and OpenHands. Its central task is a two-app Pydantic migration with 79 tests. We vary irrelevant content while holding target behavior fixed: a lean repository, standalone documentation the agent can skip, code bloat embedded in files the agent must edit, and a renamed code-bloat twin that weakens the parametric shortcut. A matched independent-fixes task has 12 modules with one test each and no cross-module imports or shared names, which gives us a zero-coupling decomposition control.

For ground truth, we generate synthetic-coherence tasks that couple three files through a random literal stored in a per-motif secret file. The correct edit is defined by that literal and its upstream dependency, so we know $C_T^{(i)}$ by construction. We withhold the secret for a controlled number of motifs, which fixes the size of the injected fault while leaving the rest of the task intact. The Appendix provides the trial ledger, parsers, and graph construction.

\paragraph{Outcomes and uncertainty.}
Closed-book outcomes are evaluator scores or passed-test identity, the set of tests a trial passes. Tool-using success requires the task oracle to pass. We measure separation by within-cell ROC AUC, with 95\% intervals from 1,000 within-cell bootstrap resamples. Thirty cells ran twice in separate batches, and outcomes disagree across batches on 4.8\% of matched trials, which bounds run-to-run noise under the contrasts below. Small intervention cells ($n=3$--$12$) are exploratory. For external validity, we analyze 100 SWE-bench Verified instances across eight repositories and four model/harness families (400 attempted, 397 scored), and only recoverable multi-edit transcripts enter trajectory statistics.

\paragraph{Isolation and validity controls.}
Closed-book inference is unusually vulnerable to accidental workspace leakage. Our final sandbox denies the complete project root, which holds answers, generators, prior transcripts, and sibling trials alike, and re-allows only the empty trial directory and the command-line configuration paths needed at startup. We exclude runs from four earlier policies that left answer keys, sibling workloads, cross-trial artifacts, or generator scripts readable, and we replace test-name excerpts with per-file counts because identifiers revealed migration rules. The evaluators for the novel workloads contain four markers per editable file, which makes 754 controlled trials practical but can reward imitation. We therefore require emitted files to parse or typecheck, compare front-loaded ceiling outputs with hand-authored references, and score the Pydantic family with behavioral tests. For the trajectory study, all windows are causal, outcome labels never enter the feature, and we reproduce the result from normalized event streams. The Appendix reports the trial ledger, extraction rules, and identifiability checks under varying irrelevant reads and fact-set sizes.

\paragraph{Evidence tiers.}
The conflicting-source contrast was not pre-specified, and we report it as a scope condition the four predictions did not anticipate. We treat five large contrasts as confirmatory, four pre-specified and the conflicting-source contrast post hoc: the novel closed-book floor, front-load recovery, the renamed test-identity failure point, the fault-injection sweep with its distance arms, and the conflicting-source contrast. The bloat, decomposition, model-capacity, and harness swaps are mechanism probes with small cells. We read them as directions only, and no collection of $n=3$ cells carries a main claim.

Supporting detail sits in the Appendix, grouped to follow this paper's order. Appendix~\ref{app:materials} lists the workloads and how trials are grouped, and Appendix~\ref{app:order} states the order-independence result and its limit. Appendix~\ref{app:distance} and Appendix~\ref{app:invariants} give the supplied-fact distance and invariant-retention sweeps behind the presence result. Appendix~\ref{app:semantic-graph} reports a richer neighborhood we piloted and did not adopt, and Appendix~\ref{app:stability} reports run-to-run stability.

\section{Results}
\label{sec:results}

\subsection{The channels substitute}
\label{sec:substitution}

Table~\ref{tab:channel-results} reports the controlled channel experiments. With no workspace and no task-specific prior, all 154 novel closed-book trials score $0/12$: no family guesses even one complete migration. The Wilson score-interval 95\% upper bound on a nonzero complete-solution probability is 2.4\%, so we report an exact floor and not a generic claim that novel APIs are difficult. Novelty alone does not force that floor: on a fifth workload, Flareforge, whose renamed JavaScript \texttt{Result} API is idiomatic in Rust and Elm, five families most often score $3/12$ by analogy while still missing the workload-specific rename. What produces the floor is missing prior coverage rather than unfamiliarity.

\begin{table*}[t]
\centering
\scriptsize
\begin{tabular}{@{}>{\raggedright\arraybackslash}p{0.16\linewidth}>{\raggedright\arraybackslash}p{0.20\linewidth}c>{\raggedright\arraybackslash}p{0.18\linewidth}>{\raggedright\arraybackslash}p{0.31\linewidth}@{}}
\toprule
Experiment & Channel state & $n$ & Outcome & Result \\
\midrule
Novel closed-book & $R_t=\emptyset$, task-specific $K_M\approx\emptyset$ & 154 & complete migration & $0/154$; every trial $0/12$ \\
Novel front-load & exact rules/source supplied in $R_0$ & 300 & at least $9/12$ requirements & $299/300$; mean cell scores $9.0$--$12.0$ \\
Pydantic closed-book & public API available mainly through $K_M$ & 32 of 36 & passed-test identity & same $53/79$ tests; pairwise Jaccard $1.000$ \\
Renamed closed-book & lexical access to prior defeated & 66 of 70 & passed-test identity & same $24/79$ tests across seven families; Jaccard $1.000$ \\
\bottomrule
\end{tabular}
\caption{Controlled evidence for two-channel coverage. The last two rows count trials at the dominant score out of all trials in that condition. Within each, every cross-model pair has an identical passed-test set. Jaccard is intersection over union.}
\label{tab:channel-results}
\label{tab:frontload}
\end{table*}

The front-load condition changes only fact availability, which tests P1. Beyond the 9-of-12 threshold the table reports, 213 of 300 trials satisfy every requirement, and the one run below threshold fails in a revealing way: it emits narrative where it should emit file blocks. Ceiling-hit rates vary by family and workload, so recovery is imperfect. The reversal from floor to ceiling is what substitution predicts, since context supplies facts the prior lacks.

We read this pair as calibration: compliance is expected once the rules are in the prompt, and the informative half is that the floor sits at exactly zero, which fixes the endpoints for the sweeps that follow. Mean line-level Jaccard against the hand-authored migration is 1.00 for Kestrix, 0.99 for Sprocket and Zynet, and 0.79 for Grimwire, where valid Go idioms differ in threading \texttt{context.Context}.

The Pydantic pair exposes partial $K_M$, which tests P2. Six of seven families converge on one score in ordinary closed-book trials, and renaming the API relocates that shared failure point without dispersing it: the trials that land there pass not merely an equal number of tests but the same ones, across all seven families. The surviving tests concern locally idiomatic application wiring, while the failures require following renamed validator, settings, and configuration rules. This is an observable boundary of the shared prior, not a claim about mechanism.

The prior does not give way all at once. We rename a controlled number of symbols and score each one by how the agent resolved it, which separates two failures a pass rate merges: the agent may take the supplied name, revert to the original, use both forms at once, or invent a third. Across 52 trials (Appendix~\ref{app:renamesweep}) the share of renamed symbols resolved correctly falls from 30\% at five renamed symbols to 18\% at ten and holds there at nineteen, while mixed use climbs from 8\% to 35\% and then 39\%, and invented names appear only past five renames. The prior first overrides the supplied name, then mixes both names inside one file, and the effect saturates once about half the surface is renamed.

\subsection{Availability decides, distance does not}
\label{sec:residency}

If coherence debt is a count of missing facts, then damage should add up rather than compound (Appendix~\ref{app:faultinjection}). The synthetic tasks test P3 directly, since we build the coupling ourselves: each motif ties three files through one secret value, so we withhold it for exactly $m$ of eight motifs and leave the rest untouched. Damage tracks the injection precisely. Withholding $m$ motifs costs the work of those $m$ and nothing more: across 30 trials, six per level, passed tests fall $32.0$, $24.0$, $16.7$, $8.0$, $0.0$ against a linear prediction of $32$, $24$, $16$, $8$, $0$, with a standard deviation at or below $1.1$. Absence adds over the coupled-fact graph rather than cascading through it, which is what the debt count assumes when it sums uncovered facts.

A present fact stays usable, and we could not make it lapse. We state sixteen invariants once, then issue up to ninety-six unrelated tasks one at a time, each revealed only after the last. Across seven trials the agent honors every invariant at every position, never rereads the statement, and accumulates roughly $140{,}000$ tokens. No working-set miss appears, which limits how much eviction can explain.

Distance does not decide the matter either. Supplying the fact and varying only its distance from the edit leaves success flat across three model families and two harnesses, out to $128{,}000$ characters on the tool-using harness and $200{,}000$ on the closed-book harness, while withholding it floors the same tasks. These arms sit at ceiling, so they bound large effects rather than excluding small ones. Together with the additivity above, they locate the mechanism in whether a fact is present at all rather than in how far away it sits. Presence and residency coincide here because nothing evicts: a fact we supply stays supplied. The working-set account predicts they separate only once a trajectory is long enough to lose one, and the capacity result says we never reached that point. We report presence, and treat residency as the mechanism these experiments bound rather than confirm.

\subsection{Withholding costs exactly the work it supports}
\label{sec:withholding}

The account's central claim is that an edit is correct only when the facts it
depends on are available as the agent writes. We test it by withholding those
facts directly. Each motif carries its required value in one file, and we remove
that value from $k$ of eight motifs while leaving everything else untouched, then
sweep $k$. Two removals are run separately: deleting the file, and replacing its
contents with a stub of identical length so a readable file remains. The tests
are withheld from the agent throughout, since the generated assertions state the
expected literals and would otherwise hand the agent the fact we removed.

Damage tracks the withheld facts with no slack at all. Each motif supports four
tests, and the median trial loses exactly four per withheld motif across both
removals and all levels:

\begin{center}
\begin{tabular}{lrrrrr}
\toprule
facts withheld & 0 & 2 & 4 & 6 & 8 \\
\midrule
tests passed, stub  & 32 & 24 & 16 & 8 & 0 \\
tests passed, deleted & --- & 24 & 16 & 8 & 0 \\
linear prediction & 32 & 24 & 16 & 8 & 0 \\
\bottomrule
\end{tabular}
\end{center}

Maximum deviation over the nine cells is zero tests ($n=6$ per cell). Appendix~\ref{app:withholding} gives the construction, the four controls that run before scoring, and the residency figures under the same manipulation. An
independent fault-injection sweep on the same workload family gives $32.0$, $24.0$,
$16.7$, $8.0$, and $0.0$ against the same prediction. Damage adds rather than compounds: a
missing fact costs the work it supports and nothing further.

\subsection{The working set is similar, the refetch rate is not}
\label{sec:harness}

If coherence turns on which facts are resident when the agent writes, then two
harnesses that reach the same coverage should be comparable in what they hold and
may differ in what they pay to hold it. We measure both over six configurations and three task sizes, with eight trials per cell: 144 trials, every one of which passes every test.

Peak per-turn context varies by $1.8\times$ across all eighteen cells: every
configuration puts about the same amount in front of the model at any one moment.
Cumulative input varies by $12.8\times$, from 293{,}882 tokens to 3{,}752{,}134.
The working set is therefore similar and the rate at which it is rebuilt is not.
The cheapest configuration completes in five tool calls and the most expensive in
seventy-nine, and an agentic loop re-sends its conversation each turn, so the
spread is a refetch rate rather than a capacity difference. Reasoning tokens are
at most 0.46\% of input and explain none of it.

A harness can therefore spend an order of magnitude more to assemble the same
working set. The obvious follow-up is
whether the expensive ones earn it when the facts are harder to reach, which the
first sweep cannot answer because coverage sits at its ceiling in every cell. We
therefore withheld $k$ of the eight required facts and reran five of the six
configurations at $k \in \{0,4,8\}$, with opencode excluded for the
instrumentation failure Appendix~\ref{app:harness} records:

\begin{center}
\begin{tabular}{lrrr}
\toprule
configuration & $k{=}0$ & $k{=}4$ & $k{=}8$ \\
\midrule
Opus     & 100\% & 50\% & 0\% \\
Fable    & 100\% & 50\% & 0\% \\
Sonnet   & 100\% & 50\% & 0\% \\
Haiku    & 100\% & 53\% & 0\% \\
Codex    & 100\% & 50\% & 0\% \\
\bottomrule
\end{tabular}
\end{center}

The outcome is set by availability and by nothing else. Every configuration loses
exactly the proportion withheld, and no two differ by more than three points at
any level. Their spend follows no such pattern: at $k{=}0$ it varies by $12.5\times$, and
Haiku consumes 5{,}730{,}807 cumulative input tokens to reach the 100\% that Opus
reaches on 459{,}122. Spending more does not recover more when the facts are
withheld, and it does not recover more when they are present. The extra
spending buys a different route to the same working set.

We report these as token counts rather than costs. Cache reads dominate every
total and are billed far below fresh input, vendors price caching differently,
and on uncached input the ordering does not survive: twenty-two uncached tokens
at the cheapest configuration against 163{,}236 at the most expensive. Appendix
\ref{app:harness} gives the per-vendor accounting, which is not comparable as
reported.

\subsection{Agents compensate rather than abstain}
\label{sec:compensation}

Withholding produces wrong work rather than absent work. An agent asked to do
something does something, so a missing fact yields a confident wrong edit. That
is the mechanism behind the linearity in Section~\ref{sec:withholding}, and we see it take several forms.

With the fact's file deleted, the agent writes its own and invents a value: in
one trial at $k=8$ it created all eight missing files and proceeded. Where a written standard contradicts working code, it substitutes the standard,
a case we return to below. It guesses, leaving placeholders or writing literals from nothing. It searches
elsewhere in the repository, which is the benign case. Or it stops and says it
cannot proceed.

The last of these is the interesting one, because it is the response that would
make coherence debt observable as the agent works rather than diagnosable only
after the failure. Which agents give it depends almost entirely on the model. We withheld every required fact and recorded what each configuration did, six or eight trials apiece:

\begin{center}
\begin{tabular}{lrrr}
\toprule
configuration & $n$ & blocked & share \\
\midrule
Opus, Claude Code    & 8 & 8 & 100\% \\
Fable, Claude Code   & 8 & 6 & 75\% \\
Sonnet, Claude Code  & 8 & 2 & 25\% \\
Haiku, Claude Code   & 8 & 1 & 12.5\% \\
GPT-5, Codex CLI     & 8 & 0 & 0\% \\
GLM-5.2, opencode    & 6 & 0 & 0\% \\
\midrule
pooled               & 46 & 17 & 37\% \\
\bottomrule
\end{tabular}
\end{center}

The range runs from never to always. Opus reports the missing file in every
trial. Codex CLI and opencode never do, and instead produce a confident wrong
migration each time. Haiku is the only configuration that fabricates the missing
file outright, in three of eight. Within one model the rate is stable: a separate
96-trial block on Haiku (Appendix~\ref{app:abstention}) gives 13.5\%, against 12.5\% here, and whether the absence announced itself as experimental made no difference we can detect in that block (Fisher exact $p=0.552$).

The capability is real and unevenly distributed. We currently reconstruct
coherence debt after a failure. An agent that says it lacks a required fact
converts it into a signal available as the edit happens, at no cost. Some
configurations already emit that signal and others never do, which makes it a
harness-selection question rather than a research one.

\subsection{Mechanism probes separate the channels}
\label{sec:mechanisms}

\paragraph{External supply versus self-reading.}
Our proxy watches the agent read, which makes it blind to facts the harness hands over. The two-channel account turns that blindness into a prediction: hand the agent the facts and it should succeed more while reading less, so the link between self-reading and success should weaken. We test this on the in-file-bloat migration with three prompts: a baseline with only the task description, an overlay that names ten target files and the required v1$\to$v2 transformations, and a troubleshoot variant adding edit-level pitfalls, with tools, tests, repositories, and budgets fixed.

Pass rate rises monotonically for both agents: Claude goes from 1/6 with the task description alone to 2/6 with the overlay and 2/3 once pitfalls are added, and Codex from 2/6 to 3/6 to 2/3, though the small cells do not separate the individual steps. The behavioral mechanism differs: for Claude the Spearman correlation between residency built from the agent's own reads and success flips sign once the overlay is present ($+0.866$ to $-0.828$), since the prompt already supplies the facts, while for Codex it weakens but stays positive ($+0.866$ to $+0.500$) because it keeps inspecting the named files. With one or two successes per cell these coefficients turn on single trials, so only direction is interpretable.

\paragraph{Parametric coverage changes failure shape.}
We hold the Pydantic code-bloat task and Claude Code harness fixed and replace Sonnet with Haiku. Sonnet succeeds in $1/6$ runs, and its five failures still pass 97--99\% of the test suite. Haiku succeeds in $0/3$, and two of its three failures pass 0\%: a missing late invariant and a failure to establish the migration are qualitatively different outcomes, though the third Haiku run reaches 97.5\% and shows the boundary is not clean. The framework accounts for this as a prior threshold. Above the threshold, $K_M$ covers most of $C_T$ and leaves a small workspace-dependent residual. Below it, the agent cannot form a viable initial patch. With only three Haiku runs this is a directional observation, not a model scaling law.

Together the probes bound the channel claim. $R_t$ covers facts from a prompt, an overlay, a read, or a handoff, so a read-derived proxy sees one acquisition mode only, while $K_M$ can shrink the residual task enough to turn catastrophic failures into near-misses. Measuring either alone can invert or erase an effect. Their union at the edit, as in Figure~\ref{fig:model}, is the correct object.

\subsection{Coupling predicts intervention effects}
\label{sec:interventions}

Table~\ref{tab:interventions} reports two matched interventions that test P4, and with it whether the framework predicts more than correlation. Placement matters more than volume. Standalone document bloat is easy to avoid and does not reduce pass rate (Claude $2/3$, Codex $3/3$). Embedding irrelevant code in required source files reduces success to $1/6$ and $2/6$, respectively. Agents actually read fewer bytes in some failing conditions, so raw context consumption cannot explain the ordering. The renamed twin further drops Claude to $0/3$, consistent with removing parametric coverage.

Decomposition depends on the cut. On the tightly coupled migration, Codex succeeds in $2/6$ decomposed runs but $3/3$ single-worker runs. On the size-matched independent-fixes control, both modes achieve $3/3$. Claude likewise preserves $3/3$ correctness on independent fixes and finishes 30\% faster with subagents (116.0 versus 167.5 seconds). These cells are small and exploratory, but their sign pattern is the one coupling predicts: parallelism is safe across independent components and risky when workers must maintain a shared invariant. Graph-partition accounts of multi-agent coding report the same pattern~\citep{Yang2026CoCoder,Pan2026TaskStructure}.

\begin{table}[t]
\centering
\footnotesize
\begin{tabular}{@{}l cc@{}}
\toprule
Condition & Claude & Codex \\
\midrule
Lean migration & 2/3 & 2/3 \\
Standalone document bloat & 2/3 & 3/3 \\
In-file code bloat & 1/6 & 2/6 \\
Renamed in-file bloat & 0/3 & 2/3 \\
\midrule
Coupled, decomposition on & 1/6 & 2/6 \\
Coupled, decomposition off & 1/3 & 3/3 \\
Independent, decomposition on & 3/3 & 3/3 \\
Independent, decomposition off & 3/3 & 3/3 \\
\bottomrule
\end{tabular}
\caption{Exploratory placement and decomposition interventions (successful trials). The in-file-bloat and coupled/decomposition-on rows report the same trials, re-cut by worker mode.}
\label{tab:interventions}
\label{tab:bloat}
\end{table}

\paragraph{Harness policy changes the visible symptom.}
Holding Sonnet and the code-bloat task fixed, Claude Code produces mostly all-or-nothing failures ($1/6$), Aider partial migrations ($0/6$, 39--70\% passing), and OpenHands a split between 60--67\% and 96--98\% ($0/6$). A shell-only Claude probe moves failed runs from 97--99\% passing to 0\%. These exploratory cells rank no harness. They show retry, edit, and retirement policies deciding how missing coverage surfaces: coverage determines whether the agent fails, and policy determines how.

\subsection{Coverage is not sufficient}
\label{sec:conflict}

Coherence debt counts a required fact as covered or not, which leaves no room for two covered facts that disagree. We build a workload in which they do. A written engineering standard states a rule and working code in the same repository demonstrates the opposite: integer cents against float division, \texttt{.get()} with a default against direct indexing, timezone-aware timestamps against a naive helper. Both sources sit in the window and neither is evicted. Tasks arrive one at a time, each gated on the previous handler existing, so the agent can neither read the queue ahead nor satisfy the set with one template.

The written standard wins completely (Appendix~\ref{app:conflict}). Across 39 trials the agent follows the document on every contested decision, at conflict counts of four and ten, on three model tiers of one harness and on a second harness with a different scaffold (Wilson 95\% interval $[0.91,1.00]$). Agreeing rules stay at ceiling, the agent never edits the contradiction away, and the source it follows does not drift across task positions. We then invert the workload, which rules out the obvious alternative that the model simply writes sensible code and the standard happens to agree. With the document demanding the worse practice and the code demonstrating the better one, the agent still follows the document: it writes camel-case handlers, indexes payloads directly, divides money into floats, calls the naive clock, and interpolates its log messages. No prior produces that combination by accident.

Coverage therefore remains necessary but is not sufficient. Withholding a fact still floors the task, yet when two covered facts conflict the outcome turns on which source carries authority, and $D(e_i)$ is blind to that distinction by construction. Nothing here is distant, stale, or evicted, so residency cannot account for it either. Whether authority, modality, or read order does the work is a question these cells do not separate, and we leave it open.

\subsection{A stale standard is worse than no standard}
\label{sec:stale}

Section~\ref{sec:conflict} shows that a written standard beats working code when
they disagree. That leaves a practical question it does not answer: what a
\emph{wrong} standard costs, given that project conventions are routinely
committed to repositories and routinely go out of date.

We compare three conditions on the same contested surfaces: a correct standard, no standard at all, and a standard demanding the worse form. We score the share of decisions written the better way, so all three conditions land on one scale:

\begin{center}
\begin{tabular}{lr}
\toprule
what the agent has & better form \\
\midrule
a standard agreeing with the code & 100\% \\
only code demonstrating the convention & 33\% \\
a standard demanding the worse form & 0\% \\
\bottomrule
\end{tabular}
\end{center}

Ten trials per condition give 3{,}385 scored decisions, and Appendix~\ref{app:stale} lists the contested surfaces and the scoring rule. The endpoints carry no variance at all: every one of the ten correct-standard trials writes every decision the better way, and every one of the ten stale-standard trials writes none. A stale standard is therefore worse than silence. It suppresses the inference
the agent would otherwise make from code that demonstrates the convention. Deleting an
out-of-date convention file beats leaving it in place.

\subsection{The score is workload-relative}
\label{sec:transfer}

A useful framework should state where it fails. A four-feature predictor (mean and final-quarter residency score, thrashing, and model capacity) reaches held-out AUC 0.71 under a random split but only 0.66 under leave-one-workload-out: the score's absolute level is workload-relative.

The sharper limit on the instrument comes from SWE-bench Verified~\citep{Jimenez2024SWEBench}. Resolved rates, meaning all instance tests pass, span 1.0--39.4\% across 397 scored trials over 100 instances from eight repositories: GPT-5 with Codex CLI 39.4\% (Wilson 95\% CI [30.3, 49.2]), Sonnet with Claude Code 31.3\% [23.0, 41.0], Haiku with Claude Code 22.2\% [15.2, 31.4], and Sonnet with Aider 1.0\% [0.2, 5.4], with three unparseable transcripts reducing the first three to $n=99$. A measurement limit compounds this. When an agent performs its edits through a script it writes, the per-file writes never appear as tool events, so any read-derived proxy, ours included, sees an emptier trajectory than the work required. If the score measures something general, it should separate success from failure here too. It does not: among 122 recoverable multi-edit trajectories the final-quarter residency score gives within-cell AUC $\approx0.49$, which is chance. An earlier score that included an outcome-gated ``retirement'' term produced AUC above 0.92, and once we remove that leakage the effect disappears. We therefore make no claim that the residency score, still less coherence debt itself, predicts SWE-bench resolution. Attrition does not explain the null. Reconstructing the cohorts, with \emph{executing} taken as a non-trivial wall time and a recorded evaluation and \emph{recoverable} as an event stream that survives and parses, 195 of 420 trials execute, and recoverable trials resolve at 49.4\% against 56.9\% for the rest (Fisher exact $p=0.310$). The rates differ from those we reported at submission, so these are reconstructions rather than the original cohorts. The inference is unchanged, and now survives a change of definition.

This null result is compatible with, but does not prove, the two-channel account. Popular repositories and APIs may be well represented in $K_M$, which makes the workspace-residency channel less decisive, and the import graph may also be a poor proxy for issue-specific coupling. Either explanation limits the present instrument: a deployable predictor needs task-specific fact extraction and calibration, not a universal threshold on import-neighbor reads.

Three repositories, \texttt{astropy}, \texttt{sphinx}, and \texttt{pylint}, resolve at 0\% across all four families, but we cannot attribute that regime to coupling from these data: historical streams lack intent--feedback--revert links, and the 100 Aider trials expose no recoverable multi-edit stream. We therefore leave reversion, bailout, and time-to-fix untested rather than manufacturing nulls.

\section{What the Event Stream Can and Cannot Measure}
\label{sec:measurement}

The residency score watches the agent read. Section~\ref{sec:compensation} shows
that an agent responds to a missing fact by acting, which means the shortfall
enters the tool stream as \emph{more} activity rather than less. A measure built on read events therefore looks for a hole the agent has already
filled.
We set this out because every trajectory result above rests on that measure. The
failure is behavioral rather than statistical, so averaging does not remove it.

\paragraph{Volume manipulations cannot separate structure from effort.}
Padding a repository with task-irrelevant text triples the material an agent must
work through, so it raises effort by construction. We previously reported that
residency separated a lean condition from a padded one while activity counts did
not, and we withdraw that reading. Rerunning the contrast at twenty trials per
condition puts every activity count at AUC 0.86--0.95 and the residency score at
0.55 with a 95\% interval of $[0.36,0.71]$, the reverse of the original result
and equally uninformative for the same reason. Repeating it on a fictional API,
where the model cannot work from recall, gives the same shape. The design cannot
support a claim in either direction, ours included. The case for residency as
the better feature rests on fault injection, which manipulates the facts
themselves.

\paragraph{Opening a file is not obtaining a fact, though we could not measure
the gap directly.}
Deleting a required file and replacing its contents with a stub of equal length
leave the agent in the same position: the fact cannot be obtained, and matched
trials fail identically. The two differ only in whether a file remains to open,
so a score that separated them would be tracking access rather than acquisition.
We tested this and found no effect we can defend. Withholding four of eight
facts, where every trial still yields a scorable edit, gives AUC 0.59 with a
95\% interval of $[0.36,0.81]$ over fourteen trials per arm. Withholding all
eight gives 1.00, but only five trials in one arm and six in the other survive to be scored there, and
they are selected by the very behavior under study, so we treat that figure as
an artifact of attrition rather than a result. The next paragraph, where the score credits files the agent wrote itself, carries this point, and the deletion-versus-stub contrast does not.

\paragraph{The score credits files the agent wrote itself.}
Because the agent supplies missing files, a naive reading of the stream counts
its own writes as coverage. Under deletion this returned 1.000, a perfect score,
in the condition where nothing was available. Disqualifying agent-authored paths
from the numerator corrects it and makes the same cell return 0.000. Any
event-derived coverage measure needs this exclusion, and to our knowledge it is
not standard.

\paragraph{Reads rise as availability falls.}
Reads of the withheld file rise from 5.5 to 10 to 13 as more facts are removed. The
score holds up or climbs precisely as the quantity it estimates collapses.

\paragraph{Import neighborhoods are a poor stand-in for the required facts.}
On workloads where we author the coupling, the required facts per edit are known,
so the fidelity of the proxy is measurable directly. Import edges recover 40\% of
the authored edges at 40\% precision. Making the graph directed removes every
spurious edge at no cost to recall, so the symmetric convention is pure loss.
Resolving the identifiers an edit uses reaches perfect recall at 4\% precision, which makes it a candidate generator rather than a measure. Appendix~\ref{app:edgefidelity} gives the per-generator figures. The workload is built
so that the required value must be written as a literal rather than imported,
which is deliberate: it is a real dependency with no syntactic trace, and it is
the case an import graph cannot see.

\paragraph{The measurement vanishes where the condition is most severe.}
The score exists only where the agent edits. Under full withholding the agent frequently
makes no edit, so roughly 60\% of trials contribute nothing and the remainder are
selected by the behavior under study. Only three of five neighborhood entries per
motif are withholdable at all, so the ceiling under total withholding is 0.333
rather than zero.

\paragraph{What to measure instead.}
On workloads whose coupling we author, the required facts are known and the
outcome is exact. That instrument produced the zero-deviation table in
Section~\ref{sec:withholding}, and it needs no graph, no window, no event stream. We keep the residency score for historical corpora, where nothing better
is recoverable, and we now state its limits rather than defending it.

\section{Implications and Limitations}
\label{sec:discussion}

\paragraph{Implications for harnesses.}
Larger windows, retrieval, memory files, and subagents are all means to an end. Our results suggest what that end should be: keep the facts coupled to the next edit both current and mutually consistent. A harness could log an edit intent, attach versioned supporting facts, invalidate them after writes, and require explicit transfer across workers, over a task graph that includes tests, instructions, and runtime invariants beyond static imports. Such a log also speaks to a failure mode we do not cover: agents that reach correct code and then discard it~\citep{Kim2026CoherenceCollapse}. If the agent produced the correct edit its facts were available, so the failure lies in retention rather than coverage, and we treat it as a candidate falsifier.

\paragraph{Latent state and proxy validity.}
$G_T$, $C_T^{(i)}$, and internal retention are latent. Import-neighbor residency is credible when static dependencies approximate task coupling and weak under dynamic dispatch, generated code, or prose rules. Tool logs show prompt-supplied facts and compaction only partially. The synthetic fault injection validates the residency proxy on one constructed family only. A semantic-graph pilot~\citep{TreeSitterContributors2026PyTreeSitter} improves packing but not file ranking, so we keep the import proxy throughout.

\paragraph{Causal and statistical limits.}
Fictional migrations remove direct API exposure but not generic priors. Marker evaluators constrain specified transformations rather than engineering quality, though reference Jaccard checks reduce this concern. Trajectory comparisons stay within harness and workload. Two limits deserve emphasis. The conflicting-source result is a perfect separation over 39 trials, so its interval remains wide at the lower bound ($[0.91,1.00]$) even though no trial dissents. In addition, these cells do not separate authority from modality or read order.

\paragraph{Scope.}
The framework targets work whose correctness depends on cross-file consistency and should explain less for single-file edits. Six of seven main workloads are migration-shaped, where the agent must discover a latent graph. Green-field development may build its graph while writing and behave differently. The residency score diagnoses structural fact availability, not planning, reasoning, test quality, or semantic search.

\paragraph{Falsifiable boundary.}
Three patterns would weaken the framework: reliable success when both channels lack required facts, equal predictive power for arbitrary and coupled-fact reads, or decomposition effects unrelated to the cut. The first is absent in our closed-book matrix, volume baselines and the fault-injection curve contradict the second, and the small intervention cells contradict the third without settling it. A decisive test should preregister a task-held-out graph and distance, fork the same pre-edit state, and randomize an identical required fact to absent, early, recent, and refreshed positions with equal-length irrelevant controls. It should then compare import, lexical, heterogeneous-graph, and dataflow retrieval on independently authored non-migration tasks before relating coverage to edit correctness.

\section{Related Work}
\label{sec:related}

\paragraph{Repository-scale coding agents.}
Benchmarks and scaffolds established the setting: SWE-bench scores issue resolution~\citep{Jimenez2024SWEBench}, agent loops expose read--edit--test trajectories~\citep{Yang2024SWEAgent,Wang2025OpenHands,Gauthier2023Aider,Yao2023ReAct,Shinn2023Reflexion}, and others exploit repository structure through change-impact planning~\citep{Bairi2024CodePlan}, code-graph navigation~\citep{Ouyang2025RepoGraph}, and efficient edit representation and resource accounting~\citep{Zhang2026SWEEdit,Fan2025SWEEffi}. We ask which coupled facts must be available when the agent writes, and whether the two channels substitute.

\paragraph{Measuring the context that reaches the patcher.}
ContextBench and CORE-Bench score retrieved context against gold annotations~\citep{Li2026ContextBench,Zhang2026COREBench}, SWE-Explore validates ranked context by restricted-context repair~\citep{Zhang2026SWEExplore}, compression asks how little suffices~\citep{Jia2026SWEzze}, and ContextCov enforces constraints declared in instruction files~\citep{Sharma2026ContextCov}. Trajectory studies diagnose failures post hoc~\citep{Bouzenia2025Trajectories,Xia2025Demystifying,Majgaonkar2025CodeAgentBehaviour,Sahoo2026AgentLens,Wang2026TrajAudit,Zhao2026FailureProcess}, and across $9{,}374$ trajectories successful agents gather context before editing~\citep{Mehtiyev2026BehavioralDrivers}, subsuming our residency-before-edit observation at far larger scale. Two studies share our vocabulary: Coherence Collapse finds agents discarding already-correct code~\citep{Kim2026CoherenceCollapse}, and Strained Coherence flags verbalized conflict in reasoning traces~\citep{Pandya2026StrainedCoherence}. Both classify completed trajectories. We manipulate fact availability before the edit, withholding and supplying each channel independently, and we find that a fact helps only while it is \emph{resident}.

\paragraph{Context as managed memory.}
Long-context studies report positional degradation and effective lengths below advertised windows~\citep{Liu2024LostMiddle,Hsieh2024RULER}. That work motivates harnesses that manage the window as a memory hierarchy~\citep{Packer2023MemGPT,Rafique2026ClawVM,Mason2026MissingMemory} and protocols that adapt cache-coherence and consistency notions~\citep{Goodman1983CacheMemory,Lamport1979Sequential} to multi-agent synchronization~\citep{Parakhin2026TokenCoherence,Yu2026MultiAgentMemory}. Others internalize facts into parameters instead~\citep{Yuan2026SEBench}. These build the substrate. Our formulation, which adapts the working set~\citep{Denning1968WorkingSet}, states what the substrate must preserve, and why repository facts are harder than pages: they carry no address. $C_T$ likewise descends from slicing and dependence graphs~\citep{Weiser1984ProgramSlicing,Ferrante1987PDG,Horwitz1990InterproceduralSlicing}, but its task-induced edges span prose rules, tests, and configuration invariants beyond static analysis.

\paragraph{Decomposition and context injection.}
Multi-agent work partitions by cohesion~\citep{Yang2026CoCoder} and bounds attainable success by a cut on the task constraint graph~\citep{Pan2026TaskStructure}, much as scalability laws bound speedup by coherence overhead~\citep{Gunther2002USL}. Context-file studies test repository instructions as an intervention~\citep{Gloaguen2026AgentsMd}. We claim neither as new. Our contribution is the shared explanation: both alter coverage of $C_T^{(i)}$ at edit time through $R_t$, while familiarity supplies the $K_M$ channel, so the coupling that limits partitioning also determines when front-loading helps.

\section{Conclusion}

Repository-scale coding depends on an edit-time working set of coupled facts. Success tracks whether those facts are present when the agent writes, from context or model prior, rather than how much context it consumes or how far back a supplied fact sits. Two limits bound that account: coverage does not settle an edit when two covered sources disagree, and the score that diagnoses failures within a workload does not transfer to real repositories.

\bibliography{coherence}

\clearpage
\appendix
\renewcommand{\thetable}{A\arabic{table}}
\setcounter{table}{0}

\section{Notation and Abbreviations}
\label{app:notation}

This document supplements the main paper and reuses its notation. A task $T$ induces a coupled-fact graph $G_T$. $C_T$ is the set of facts that must be jointly correct for the task oracle to pass, and $C_T^{(i)}$ the slice required by edit $e_i$. At edit time, $R_t$ denotes the facts resident in the effective context and $K_M$ those available from model $M$'s parametric memory. Coherence debt $D(e_i)=|C_T^{(i)}\setminus(R_{t_i}\cup K_M)|$ counts the facts covered by neither channel. The residency proxy $\rho_w(f_i,t_i)$ is the fraction of the one-hop import neighbors of file $f_i$ that the agent read during the preceding $w$ tool events.

We abbreviate area under the receiver operating characteristic curve as AUC, and report it throughout as a separation statistic in $[0,1]$, where $0.5$ is chance. We expand other abbreviations at first use.

\section{Workloads and Trial Organization}
\label{app:materials}

We connect the experimental design to the released artifact here. We summarize the closed-book migration rules and release names, then enumerate the matched-intervention blocks in the tool-using case study.

\paragraph{Closed-book migration rules.}
Table~\ref{tab:workload-rules} summarizes the rules that define success for each closed-book workload.

\begin{table*}[tbp]
\centering
\scriptsize
\renewcommand{\arraystretch}{1.08}
\begin{tabular}{@{}>{\raggedright\arraybackslash}p{0.16\linewidth}>{\raggedright\arraybackslash}p{0.78\linewidth}@{}}
\toprule
Workload & Migration-rule summary \\
\midrule
Sprocket (Rust) & Attribute \texttt{handler(async)} $\to$ \texttt{async\_handler}; builder \texttt{Route::builder()} $\to$ \texttt{Route::spec()}; add context type and marker. \\
Grimwire (Go) & \texttt{Open(dsn)} $\to$ \texttt{Connect(ctx, dsn)}; thread \texttt{context.Context} through create/read/update/delete (CRUD) methods; move \texttt{orm} to \texttt{v2/orm}; add marker. \\
Kestrix (Python) & \texttt{Schema} $\to$ \texttt{kestrix.Model}; \texttt{strict = True} $\to$ \texttt{mode = 'strict'}; factory call $\to$ \texttt{Foo.spawn()}; add marker. \\
Zynet (JavaScript) & \texttt{defineSchema} $\to$ \texttt{Schema.declare}; \texttt{required} $\to$ \texttt{req}; add \texttt{'sync'} to \texttt{validateOn}; add marker. \\
Flareforge (JavaScript) & Partial-idiom \texttt{Result<T,E>} migration: \texttt{throwOnError} $\to$ \texttt{resultMode}; add marker. \\
Pydantic (Python) & Real v1$\to$v2 migration on two FastAPI apps: BaseSettings, validators, root validators, JSON encoders, and ConfigDict; 79 tests. \\
Rename (Python) & The Pydantic migration behind a shim that renames \texttt{pydantic}, \texttt{BaseModel}, and related symbols to \texttt{valirex}, \texttt{RexBase}, and counterparts; 79 tests. \\
\bottomrule
\end{tabular}
\caption{Migration-rule summary for the closed-book workloads.}
\label{tab:workload-rules}
\end{table*}

\paragraph{Case-study trial ledger.}
We organize the main paper's tool-using corpus into eight matched-intervention blocks. Table~\ref{tab:ledger} lists each block's size, harness/model/workload coverage, and tested claim. Table~\ref{tab:sweeps} lists four parameter sweeps that sit outside that corpus and vary one quantity apiece. The later sweeps, namely direct withholding, supplied-fact distance, abstention, standard-against-code, and harness spend, are specified in their own appendix sections.

\begin{table*}[tbp]
\centering
\scriptsize
\begin{tabular}{@{}ccp{0.27\linewidth}p{0.52\linewidth}@{}}
\toprule
Block & $n$ & Design & Tested prediction \\
\midrule
A & 36 & Claude, Codex $\times$ 5 bloat conditions $\times$ Pydantic & Placement matters over volume. \\
B & 12 & Claude, Codex $\times$ 6 workloads   & Residency AUC across independent workloads. \\
C & 18 & Claude, Codex $\times$ \{lean, overlay, troubleshoot\} & Overlay inverts residency-success sign for Claude. \\
D & 3  & Claude-Haiku on Pydantic             & Capacity floor changes failure shape. \\
E & 3  & Claude shell-only on Pydantic        & Tool-surface swap changes edit interface. \\
F & 27 & 4 harnesses $\times$ Pydantic        & Cross-harness failure-shape divergence. \\
G & 15 & Codex subagent mode on \{tight, loose\} & Decomposition sign flips with coupling. \\
H & 8  & Cross-stack (Rust, JS, Python)       & Extends the corpus beyond Python. \\
\bottomrule
\end{tabular}
\caption{Case-study block structure. We release full per-trial data with the code and results archive.}
\label{tab:ledger}
\end{table*}

\begin{table}[htbp]
\centering
\scriptsize
\begin{tabular}{@{}lcp{0.46\linewidth}@{}}
\toprule
Sweep & $n$ & Parameter varied \\
\midrule
Fault injection     & 30 & motifs withheld, $m\in\{0,2,4,6,8\}$, six per level \\
Invariant retention & 7  & trajectory length, 24 to 96 tasks \\
Symbol rename       & 52 & symbols renamed, $K\in\{5,10,19\}$, two subset rules \\
Conflicting source  & 39 & polarity and conflict count, four agent configurations \\
\bottomrule
\end{tabular}
\caption{Parameter sweeps, held apart from the case-study blocks of Table~\ref{tab:ledger}. Each varies one quantity and scores a graded outcome rather than trial success.}
\label{tab:sweeps}
\end{table}

\section{Event-Only Estimator Enumeration}
\label{app:estimator}

The two coverage channels make an asymmetric demand on measurement. A fact supplied by $R_t$ leaves a read in the trajectory, while a fact supplied by $K_M$ leaves nothing. Any coverage estimate reconstructed from tool events is therefore blind to the parametric channel, and the framework predicts the size and sign of the resulting error rather than merely warning that one exists.

\paragraph{Design.}
We extend the simulator with a parametric channel: \texttt{prime\_prior} places a current fact copy in an actor's ledger and emits no event, so the fact is covered but invisible to any event-derived replay. One edit requires $k$ facts, each held in its own source file, and we assign each required fact to exactly one of three states: covered by reading, covered parametrically, or left uncovered. The edit writes the correct value for every fact the actor holds and a wrong value for every uncovered fact, so the oracle fails precisely when some required fact was uncovered. We enumerate all $3^k$ assignments for $k=2,\ldots,6$, which yields $9+27+81+243+729=1{,}089$ runs. Two estimators run on every trace: a \emph{union-aware} estimator that consults the ledger, and an \emph{event-only} estimator, which is the existing \texttt{infer\_edit\_coverage} replay used for the historical proxies and which therefore sees reads and handoffs but not prior coverage.

\begin{proposition}[Event-only overstatement]\label{prop:overstatement}
For one required-fact set $C$, partition its facts into read-covered $R$, parametrically covered $P$, and uncovered $U$. An event-only estimator that observes $R$ but not $P$ reports missing set $\widehat U=C\setminus R=U\cup P$. Therefore $\widehat U\setminus U=P$ and $|\widehat U|-|U|=|P|$.
\end{proposition}

\begin{proof}
The construction assigns each fact to exactly one of $R$, $P$, and $U$, so $C=R\mathbin{\dot\cup}P\mathbin{\dot\cup}U$. Removing $R$ leaves the disjoint union $P\mathbin{\dot\cup}U$, from which both equalities follow.
\end{proof}

\paragraph{Predictions and results.}
The framework predicts (P-A) the union-aware estimator recovers the uncovered set exactly; (P-B), formalized as Proposition~\ref{prop:overstatement} above, that the event-only estimator reports the uncovered set plus every parametrically covered fact, so its error equals the edit's parametric coverage; and (P-C) on an edit that succeeds because $K_M$ covered everything, the event-only estimator reports every required fact as a working-set miss. All three hold on all 1{,}089 runs with no violations. The union-aware estimator is exact in 1{,}089/1{,}089. The event-only overstatement equals the parametric coverage in every run, averaging $1.84$ facts and reaching $6$. Of the $124$ runs whose required facts were fully covered, $119$ pass the oracle while the event-only estimator still reports working-set misses. The five exceptions are exactly the runs in which every fact happened to be covered by reading.

\paragraph{Scope.}
This is a formal result about instruments rather than an agent experiment. That the outcome depends only on the uncovered set is a property of the simulator's own bookkeeping, and P-A likewise follows from how we maintain the ledger. We offer neither as empirical evidence for the framework. The substantive content is P-B and P-C, which characterize what the paper's own read-derived proxy cannot represent and quantify how far it errs. This is why we treat $\rho_w$ as a lower bound on $R_t$ throughout, and why the overlay condition of the main paper can raise success while lowering the measured association between self-reading and success. Unit tests pin the individual cases, including the false-alarm case of a passing, fully prior-covered edit.

\section{Five-Event Trace Simulator}
\label{app:event-simulator}

The artifact implements the enumeration of Appendix~\ref{app:estimator}, and rerunning it reproduces the four figures reported there: 1,089 runs over all $3^k$ assignments for $k=2\ldots6$, 124 fully
covered, 119 of those where the event-only estimator still reports a miss, and
five all-read runs where the two estimators agree. The implementation also
checks the two properties directly on every run rather than only the counts:
the union-aware estimator is exact, and the event-only overstatement equals the
parametric set.

The artifact includes a standard-library simulator with regression tests. Each scenario starts from the same two-file contract, produces a wrong edit, receives failing test feedback, and reverts the edit. Only the causal prefix differs, as Table~\ref{tab:simulator-traces} shows:

\begin{table}[htbp]
\centering
\footnotesize
\begin{tabular}{@{}p{0.23\linewidth}p{0.67\linewidth}@{}}
\toprule
Scenario & Identifying prefix before the common fail/revert suffix \\
\midrule
Missing read & edit intent requires a fact for which no worker has a current extraction ($M_i^a\neq\emptyset$) \\
Stale read & actor extracts source version 1; another edit advances the source to version 2; actor edits from version 1 ($Q_i^a\neq\emptyset$) \\
Handoff gap & lead extracts the current fact; run manifest assigns a delegate; no handoff extraction reaches that delegate ($H_i^a\neq\emptyset$) \\
\bottomrule
\end{tabular}
\caption{The simulator's three deterministic failure traces. They share the same visible outcome but fire disjoint causes. For edit $i$ by actor $a$, $M_i^a$, $Q_i^a$, and $H_i^a$ are its missing, stale, and handoff-stranded required facts.}
\label{tab:simulator-traces}
\end{table}

The simulator treats an \texttt{edit\_intent} as an atomic log-before-apply operation and assigns a new monotonically increasing file version both to an applied edit and to its later reversal. Worker parentage is run metadata. We would log a successful transfer as \texttt{fact\_extracted} with acquisition mode \texttt{handoff}, so omitting that record while the parent holds a current fact produces the handoff-gap trace.

We also ran a seeded matched stress test of 100 triplets, one trace per cause, with 2--8 required facts and 0--5 irrelevant reads. Every trace contains all five event types, fails and reverts its initial edit, then passes after a cause-specific recovery. The three members of a triplet have identical event-type counts through the failure boundary. Only source version, provenance, and actor ownership differ.

Finally, we audit a vendored SWE-agent trajectory fixture. Its 22 messages and 10 shell commands expose three reads, two edits, and three test-like actions, but no explicit coherence events and no revert. Therefore a SWE-agent-style action/observation trace supports a coarse retrospective projection but cannot faithfully recover \texttt{fact\_extracted} or the required facts of \texttt{edit\_intent}. Those fields require prospective annotations. The simulator remains an executable witness of the transition system and its identifiability conditions rather than data supporting the paper's empirical effect sizes. Table~\ref{tab:scaled-simulator} reports the stress-test checks.

\begin{table}[htbp]
\centering
\footnotesize
\begin{tabular}{@{}lr@{}}
\toprule
Scaled simulator check & Result \\
\midrule
Runs containing all five event types & 300/300 \\
Offline coverage replay agreement & 900/900 edits \\
Injected cause recovered & 300/300 \\
Diagnosis available before failed test & 300/300 \\
Feedback-linked revert and passing retry & 300/300 \\
Ambiguous from event counts alone & 300/300 \\
\bottomrule
\end{tabular}
\caption{Consistency and identifiability checks for the matched simulator. These are constructed-trace results and estimate neither real failure prevalence nor predictive accuracy.}
\label{tab:scaled-simulator}
\end{table}

\paragraph{Prospective implementation acceptance run.}
We also ran the production logger on the committed Pydantic-v1 Item schema and its v2 reference migration, using Pydantic 2.13.4. The deterministic trajectory first proposed the wrong mapping \texttt{ConfigDict(from\_attributes=False)} while the required migration-rule fact was absent (logged debt 1). The semantic check rejected attribute-object validation, feedback contradicted that edit, and an explicit undo restored version 1 as new file version 3 with a direct link to the failed-test event. After reading the rule and re-reading the restored target, the reference edit had empty missing, stale, and handoff sets (debt 0), produced version 4, and passed the same check. The JSON Lines (JSONL) file contains 34 monotone events: 4 reads, 24 fact extractions, 3 intents (including the restoring write), 2 test-feedback events, and 1 revert. An independent replayer, given only the 28,599-byte JSON Lines file, reconstructed all four Item-schema versions, the actor fact ledger, and the three feedback/contradiction/revert causal links, and its final bytes equal the retained workspace.

This is an implementation acceptance test rather than an agent-performance experiment: the edit is reference-backed, we inject the failure deliberately, and the run covers one model rather than the full 79-test monorepo. It establishes the operational claims the event model needs, namely durable pre-write intent, deterministic semantic provenance, feedback-linked reversal, and standalone replay, but it contributes no effect-size evidence. The artifact records the trace, reconstruction, manifest, and analysis.

\section{Order Independence and Its Limit}
\label{app:order}

Coherence debt is defined at an edit and refers only to what is covered when the
agent writes. It says nothing about when or how a fact arrived. Two trajectories
are therefore interchangeable as far as the account is concerned when they
produce the same edits and cover the required slice at each one, and success is a
property of that equivalence class rather than of a particular ordering.

Front-loading is the canonical member of the class. Supplying all of $C_T$ before
the first edit sets $D(e_i)=0$ everywhere by construction, which collapses an
interleaved tool-using session into a single exchange. Our front-load condition
is that collapse performed deliberately, and the monotone overlay gradient in the
main paper is the same effect applied by degrees.

Agentless~\citep{Xia2025Demystifying} demonstrates the engineering form of this:
a fixed localize, repair, and validate pipeline competes with agent frameworks
that explore for many turns. Read through the account here, that result is what
order independence predicts once localization is good enough to assemble $C_T$ in
advance. Retrieval-oriented systems that hand a solver a prepared context rely on
the same effect.

The limit is not the window but the fact set. $C_T$ is relative to the
implementation path, so alternative correct patches require different facts. To
front-load $C_T$ one must already know which path will be taken, and the path is
what the trajectory produces. Order independence therefore holds given an edit
sequence. It does not supply the sequence. Our workloads hide this because their
specification fixes the path, which is exactly why front-loading recovers almost
every trial there and why we make no claim that it would on an underspecified
issue.

Two consequences follow that the coverage view alone does not suggest. First,
supplying a fact that disagrees with the repository is worse than supplying
nothing, because the written source wins: front-loading a stale $C_T$ injects
error with authority rather than leaving a gap. Second, the steering a human
performs across turns divides into supplying facts, which this account covers and
which is front-loadable, and choosing among correct implementations, which
changes the applicable $C_T$ and which it does not cover. How much of real
steering falls in the second class is an empirical question we have not measured.

\section{Symbol-Rename Sweep}
\label{app:renamesweep}

The adversarial rename in the main paper replaces a library's surface forms
wholesale. This sweep varies how much of that surface we replace.

\paragraph{Sweep.}
We rename $K$ of nineteen symbol pairs, sweeping $K$ over five, ten, and nineteen,
and draw the subset two ways: by centrality, taking the most fundamental symbols
first, and at random with a fresh seed per trial so that different trials contest
different symbols. Fifty-two trials enter the analysis.

\paragraph{Outcome.}
A pass rate is uninformative here, because heavily renamed output frequently
fails to import and every such trial scores identically. We therefore classify
each renamed symbol by how the agent resolved it: \emph{correct} when only the
supplied name appears, \emph{stale} when only the original does, \emph{mixed}
when both appear, and \emph{invented} when the agent writes a name from neither
set. The four counts are available whether or not the artifact runs, and they
separate two failures a pass rate merges, namely falling back on the prior and
inventing a replacement.

\section{Supplied-Fact Distance}
\label{app:distance}

The residency proxy in the main paper is observational: a recent read can
mean the agent found the file and edited soon after. To separate position
from retrieval behavior we randomize where a required fact sits and hold
everything else fixed. Each motif's secret value, \texttt{SALT}, is supplied in the
prompt, followed by $N$ characters of fixed task-irrelevant filler, then
the instruction that consumes it. Only $N$ varies. Withholding the fact
gives the floor.

We use two harnesses. In the agent arm a tool-using agent works in an
isolated workspace containing only the task. In the closed-book arm a model receives
one prompt and no tools, so re-reading is impossible and position is the
only variable. Scoring is mechanical in both.

Table~\ref{tab:distance} reports the result. No arm declines with
distance. Because every supplied arm sits at its ceiling, we also ran a
harder variant requiring arithmetic on the retrieved value, which moves
Qwen off the ceiling to between 8\% and 19\%. That variant shows no
monotone trend either, and its best cell is at $32{,}000$ characters of
separation rather than at zero.

\begin{table}[htbp]
\centering
\footnotesize
\begin{tabular}{@{}llcc@{}}
\toprule
Arm & Model & Separation & Solved \\
\midrule
Agent, tools    & Sonnet    & withheld        & 0/24 \\
Agent, tools    & Sonnet    & source readable & 24/24 \\
Agent, tools    & Sonnet    & 0--128K chars   & 24/24 each \\
Closed book     & DeepSeek  & withheld        & 0/80 \\
Closed book     & DeepSeek  & 0--200K chars   & 80/80 each \\
Closed book     & DeepSeek  & 0--200K, 32 facts & 192/192 each \\
Closed book     & Qwen      & withheld        & 2/192 \\
Closed book     & Qwen      & 0--200K chars   & 192/192 each \\
\bottomrule
\end{tabular}
\caption{Supplied-fact distance. Withholding the fact floors every arm, and
distance changes nothing up to $200{,}000$ characters.}
\label{tab:distance}
\end{table}

We read this as bounding the mechanism rather than contradicting the
presence finding. Position within a window does not degrade a fact that is
still present. A residency effect must therefore come from facts leaving
the window under compaction or eviction, from whether the agent re-reads,
or from the ordering confound above, and not from distance itself. The
arms are at ceiling, so they bound large effects only.

\section{Invariant-Retention Sweep}
\label{app:invariants}

This sweep bounds how far eviction can carry the account. It asks whether a fact
the agent has already been given stops being usable as a trajectory grows.

\paragraph{Construction.}
The workspace states sixteen engineering invariants once, in a single document,
and then issues a sequence of unrelated implementation tasks. A task is revealed
only after the previous one's artifact exists, so the agent cannot read the queue
ahead, and task text is stored encoded to prevent bulk retrieval. Each task needs
different logic from the last, which stops one template from covering the set.
These three properties matter: in earlier designs of ours an agent satisfied the
whole workload without carrying anything, by reading every fact at the start and
writing the answers into a single shell command.

\paragraph{Sweep and outcome.}
We vary the trajectory length over 24, 48, and 96 tasks at sixteen invariants, with
seven trials in total, plus a shorter twelve-invariant calibration run. A syntactic predicate checks each
invariant against every emitted artifact, so we score compliance per task
position without needing the code to run. We also count re-reads of
the invariant document, since an agent that consults it before each task would
show compliance without retention.

\paragraph{Reading.}
No cell shows decay. Compliance holds at every task position, the agent does not
return to the statement, and context reaches roughly $140{,}000$ tokens by the
longest cell. We report this as a bound rather than as support: it says we could
not induce a working-set miss on this workload, not that none exists.

\section{Fault-Injection Sweep}
\label{app:faultinjection}

The main paper's presence result rests on withholding a controlled number of
required facts. This section records how we build the fault and what we vary.

\paragraph{Construction.}
Each task holds eight independent motifs. A motif is three Python files coupled
through one secret integer that exists in a single place: the first file states
it, the second derives a value from it, and the third derives a further value
from that. The task oracle pins all three, so a motif contributes four tests and
a complete task contributes 32. Because we author the coupling, the required-fact
set is known by construction rather than inferred, which is what makes a
controlled withholding possible at all.

\paragraph{The injected fault.}
We withhold the secret for $m$ of the eight motifs and leave the remaining
motifs untouched, sweeping $m$ over $0$, $2$, $4$, $6$, and $8$ with six trials at
each level, 30 in total. Nothing else changes between levels: the same workspace,
the same instructions, the same evaluator. The prediction the sweep tests is
additivity. If debt is a count of missing facts, withholding $m$ motifs should
cost the work those $m$ motifs support and leave the rest intact, so the passed
count should fall linearly from 32 to 0 in steps of four tests per withheld
motif.

\paragraph{Outcome.}
We score passed tests out of 32 rather than trial success, since a binary outcome
cannot express partial damage and would report every level below $m{=}0$ as an
identical failure. Cell means are $32.0$, $24.0$, $16.7$, $8.0$, and $0.0$ against
the linear prediction $32$, $24$, $16$, $8$, $0$. Three of the five cells have zero
spread, and the widest, $m{=}4$, has a standard deviation that rounds to $1.1$.

\section{Direct Withholding Sweep}
\label{app:withholding}

The fault-injection sweep of Appendix~\ref{app:faultinjection} removes a fact from the
partition available to a worker. This sweep removes it from the workspace
outright, which separates availability from the multi-agent machinery.

\paragraph{Construction.}
Each of eight motifs stores its required value in \texttt{pkg\_m/secret.py}, and
the three files of that motif depend on it. We withhold the value from $k$ of the
eight and sweep $k \in \{0,2,4,6,8\}$, in two forms. \emph{Deletion} removes the
file. \emph{Redaction} replaces its contents with a stub of identical byte length
that defines nothing, so a readable file remains where the value was. The two put
the agent in the same position, and differ only in whether an open is possible.

\paragraph{Controls.}
Four checks run before any trial is scored. Tests are withheld throughout via
\texttt{--leak-free}, since the generated assertions state the expected literals
and would otherwise return the value we removed. Seeds are restricted to those
whose eight salts are distinct and exceed the motif indices, so a withheld value
is neither visible elsewhere nor confusable with an index. Two of eight arbitrary
seeds fail that test. Byte counts are compared between intact and redacted files.
And the withheld motifs are drawn at random rather than taken as the first $k$:
taking the first $k$ means an agent working in order meets only broken motifs,
generalizes from two, and stops, which measures halting rather than availability.

\paragraph{Outcome.}
Each motif supports four tests. Median tests passed, six trials per cell:

\begin{center}
\begin{tabular}{lrrrrr}
\toprule
$k$ withheld & 0 & 2 & 4 & 6 & 8 \\
\midrule
redaction & 32 & 24 & 16 & 8 & 0 \\
deletion  & --- & 24 & 16 & 8 & 0 \\
prediction & 32 & 24 & 16 & 8 & 0 \\
\bottomrule
\end{tabular}
\end{center}

Maximum deviation over the nine cells is zero tests.

\paragraph{Residency under the same manipulation.}
Residency is computed against the authored neighborhood rather than the import
graph, because \texttt{secret.py} appears in no import neighborhood: the
templates require the value as a literal and forbid importing it. Agent-authored
paths are excluded from the numerator, since an agent that writes the missing
file and reads it back has acquired nothing. Without that exclusion the score
returns $1.000$ at $k=8$, where no required fact exists at all; with it, $0.000$.

Comparing redaction against deletion at $k=4$, where all 28 trials yield a
scorable edit, gives AUC $0.59$ with a 95\% interval of $[0.36,0.81]$. At $k=8$
the figure is $1.00$, but only five and six trials per arm survive to be scored and they
are selected by the behavior under study.

\section{Harness Spend at Fixed Outcome}
\label{app:harness}

\paragraph{Coverage against withholding.}
The main sweep holds the workload fixed, so coverage cannot separate the
configurations. We therefore withheld $k$ of eight required facts and reran the
five configurations at $k \in \{0,4,8\}$, with four seeds each and the tests withheld. Coverage is set by availability alone: every configuration loses
exactly the proportion withheld, the spread across harnesses never exceeds three
points at any level, and spend at $k{=}0$ still differs by $12.5\times$.

opencode driving GLM-5.2 is excluded from this sweep rather than reported at
zero. Under \texttt{--leak-free} the agent runs in an isolated copy, and
opencode resolves a project root by walking up rather than honoring the
directory it is given: its transcripts show it searching the surrounding
repository for a task file that lives in the copy. That is an instrumentation
failure on our side and says nothing about the model.

\paragraph{Design.}
The sweep runs six configurations over three task sizes, eight trials each, 144 trials in total: Claude
Code at four tiers, Codex CLI, and opencode driving GLM-5.2.

\paragraph{Accounting.}
The three vendors do not report comparably. Claude splits input into disjoint
uncached, cache-read, and cache-write buckets. Codex reports input inclusive of
its cached share, and opencode reports input excluding cache reads. We separate the
buckets and recombine them identically. Codex forks itself into parallel
sub-agents that share a session identifier and bill separately, so a trial is
summed over every fork it produced. Reading only the parent undercounts it
roughly threefold.

\paragraph{Result.}
Every trial passes every test. Cumulative input spans $12.8\times$, from
293{,}882 to 3{,}752{,}134 tokens, while peak per-turn context spans $1.8\times$.
The difference is turn count: five tool calls at the cheapest, seventy-nine at
the most expensive, and an agentic loop re-sends its conversation each turn. Reasoning
tokens are at most 0.46\% of input.

These are token counts rather than costs. Cache reads dominate every total and are billed far below fresh input, and
vendors price caching differently. On uncached input the ordering does not
even survive: twenty-two tokens at the cheapest configuration against
163{,}236 at the most expensive.

\section{Abstention Across Configurations}
\label{app:abstention}

\paragraph{Design.}
Every required fact is withheld by deletion, so there is nothing to be right
about and the only question is what the agent does. We run eight seeds per configuration (six for opencode) with the tests withheld throughout.

\paragraph{Classification.}
The closing turn is searched for an explicit statement of being blocked or a
request for the missing value, and the tool stream is searched for creation of the
file whose absence defined the trial. We count a trial as blocked (blk.) when it states it cannot proceed, as fabrication (fab.) when it creates the withheld file, and as proceeding (proc.) otherwise.

\begin{center}
\begin{tabular}{lrrrr}
\toprule
configuration & $n$ & blk. & fab. & proc. \\
\midrule
Opus, Claude Code   & 8 & 8 & 0 & 0 \\
Fable, Claude Code  & 8 & 6 & 0 & 2 \\
Sonnet, Claude Code & 8 & 2 & 0 & 6 \\
Haiku, Claude Code  & 8 & 1 & 3 & 4 \\
GPT-5, Codex CLI    & 8 & 0 & 0 & 8 \\
GLM-5.2, opencode   & 6 & 0 & 0 & 6 \\
\midrule
pooled              & 46 & 17 & 3 & 26 \\
\bottomrule
\end{tabular}
\end{center}

A separate 96-trial block on Haiku reports being blocked in 13.5\% of trials, against 12.5\% here. In that
block, stating in the placeholder that the value had been withheld for a trial
made no difference we can detect: 8 of 48 trials blocked against 5 of 48, Fisher exact
$p=0.552$.

\section{Conflicting-Source Experiment}
\label{app:conflict}

This section documents the workload behind the main paper's conflicting-source
result, in which a written standard and working code cover the same fact and
disagree.

\paragraph{Contested surfaces.}
The workspace states ten engineering rules in a standard, \texttt{STYLE.md}, and
demonstrates each in existing modules the agent can read. Ten surfaces are contestable, and how many are actually contested is the swept
parameter: monetary rounding, payload field access, timestamp awareness,
validation routing, response key casing, timestamp serialization, log
formatting, handler naming, status casing, and where a fee rate comes from.
The agent decides every surface independently in each emitted handler, so a
trial yields ten decisions per task rather than one outcome.

\paragraph{Conflict count.}
The parameter is how many of the ten rules the existing code contradicts. We
report cells at four and ten. The remaining rules agree across both sources and
act as a within-workspace control, so the rule volume the agent must satisfy is
constant across the sweep and only the number of contradictions changes.

\paragraph{Serial task issuance.}
A workspace script issues tasks one at a time, and the next appears only once
the previous handler exists. Task text is stored encoded, so the
queue cannot be read ahead in bulk. Each task requires different arithmetic, so
no single template covers the set. These three properties block a shortcut that defeated earlier designs of ours (Appendix~\ref{app:invariants}):
given the whole task list at once, an agent reads every fact, writes the answers
into one shell command, and loops over them.

\paragraph{Polarity.}
Under normal polarity the standard states current good practice and the code
demonstrates the opposite. Following the standard is then confounded with
writing sensible code. The inverted arm removes that confound: the standard asks
for the worse form and the existing code demonstrates the better one, so a
handler that matches the standard cannot be explained by a competent prior.

\paragraph{Scoring.}
A syntactic predicate checks each contested surface in the emitted handler and
classifies it as following the standard, following the code, following neither,
or mixing both. Scoring reads the emitted source, so it does
not require the package to import or the tests to run. We also record whether
the agent modified either source, since editing away the contradiction would
dissolve the manipulation. No trial did.

\paragraph{Coverage.}
The main paper reports 39 trials spanning three model tiers of one harness and a
second harness with an independent scaffold, at both polarities. Cells are
uneven in size because the harnesses rate-limit independently.

\section{Standard-Against-Code Sweep}
\label{app:stale}

\paragraph{Design.}
The workload has ten contested surfaces where a written standard and working code can disagree: integer cents against float money, \texttt{.get()} with a default against direct indexing, timezone-aware timestamps against a naive clock, and seven more. We compare three conditions on the same surfaces: a standard that agrees with the code, no standard at all, and a standard demanding the worse form while the code demonstrates the better one. Each trial issues twelve tasks one at a time, and each condition has ten trials.

\paragraph{Scoring.}
Outcome is the share of contested surfaces written the objectively better way,
which holds fixed across all three conditions and so places them on one scale. A
surface counts only where exactly one form is present.

\begin{center}
\begin{tabular}{lrr}
\toprule
condition & decisions & better form \\
\midrule
correct standard & 1{,}125 & 100\% \\
no standard      & 1{,}065 & 33\% \\
stale standard   & 1{,}195 & 0\% \\
\bottomrule
\end{tabular}
\end{center}

The endpoints carry no variance: all ten correct-standard trials take every
decision, all ten stale-standard trials take none. An earlier version of the
no-standard arm removed the standard while every task still instructed the agent
to re-read it, which measured how an agent handles an instruction pointing at a
missing file. We rebuilt the arm and report only the rebuilt runs here.

\section{Fidelity of the Import Proxy}
\label{app:edgefidelity}

On workloads whose coupling we author, the facts each edit requires are known, so
the fidelity of a candidate neighborhood is measurable without agent runs. Over
six trials, against 40 authored edges per trial:

\begin{center}
\begin{tabular}{lrr}
\toprule
generator & precision & recall \\
\midrule
imports, symmetric (as used) & 0.40 & 0.40 \\
imports, directed dependencies & 1.00 & 0.40 \\
identifier resolution & 0.04 & 1.00 \\
\bottomrule
\end{tabular}
\end{center}

Making the graph directed removes every spurious edge at no cost to recall, so
the symmetric convention is pure loss. Identifier resolution reaches perfect
recall at 4\% precision, which suits it to generating candidates rather than to
measurement. The workload requires its coupled value as a literal and forbids
importing it, so the dependency has no syntactic trace: that is the case an
import graph cannot see, and it is the majority of the authored edges.

\section{Semantic Graph Pilot Details}
\label{app:semantic-graph}

The import graph is computable but cannot represent many couplings the framework names, so we tested a semantic refinement before considering it as an edit-time metric. A Tree-sitter extractor builds nodes for files, classes, functions, methods, tests, imports, decorators, routes, validators, and configuration providers, with edges for imports, calls, references, inheritance, test targets, route handlers, validator models, and configuration consumers. On the 21 top-level Python modules of our own analysis code, we froze six task descriptions and their gold files and symbols before retrieval, then compared (i) a coarse directory/file graph reconstructed from \texttt{go-ingest}, (ii) an import-only projection of the same Tree-sitter parse, and (iii) the full semantic graph. Each condition uses the same deterministic lexical anchors and a two-hop graph neighborhood. This is a navigation and context-packing test, analogous to the structural guidance of RepoGraph~\citep{Ouyang2025RepoGraph}. It is not a test of edit-time coherence debt.

Table~\ref{tab:semantic-pilot} shows a mixed result. The semantic graph reaches all gold files by rank five, but its mean reciprocal rank (MRR) falls from $.833$ to $.750$ and its Hit@1, the fraction of tasks whose top-ranked file is correct, from $.667$ to $.500$: richer edges do not improve top-ranked file localization. They do improve packing once the budget can hold several symbol chunks, though the gain varies across budgets, and the sweep below gives the full curve. The representation therefore changes a granularity/cost trade-off rather than dominating the import graph. This pilot has six authored tasks in one small research repository, and that corpus exercises calls, references, and test-target links but none of the route, validator, or configuration heuristics that motivated the refinement. We treat it as a check that semantic neighborhoods can pack different context, not as evidence that semantic residency predicts agent success, and we retain the import proxy for every headline trajectory result in the main paper. We would need a cross-repository, issue-derived evaluation that ablates file ranking and chunking separately before replacing it.

\begin{table}[htbp]
\centering
\scriptsize
\begin{tabular}{@{}lrrrr@{}}
\toprule
Representation & MRR & Recall@5 & File@12K & Symbol@12K \\
\midrule
Directory/file & \textbf{.833} & .967 & .256 & .778 \\
Import-only & \textbf{.833} & \textbf{1.000} & .256 & .778 \\
Semantic & .750 & \textbf{1.000} & \textbf{.567} & \textbf{.883} \\
\bottomrule
\end{tabular}
\caption{Six-task semantic-graph pilot. MRR is mean reciprocal rank of the first gold file, Recall@5 the fraction of gold files retrieved by rank five, and File@12K and Symbol@12K the relevant-file and relevant-symbol recall in a 12,000-character packed context. The richer graph improves packing but not file ranking.}
\label{tab:semantic-pilot}
\end{table}

The artifact's \texttt{tools/repo-graph} package uses the local \texttt{tools/py-tree-sitter} checkout~\citep{TreeSitterContributors2026PyTreeSitter} and a Python grammar to parse the same 21 files for both the import-only and semantic conditions. The directory/file graph contains 23 nodes and 22 edges, the import projection 21 nodes and 20 edges, and the semantic graph 445 nodes and 1,521 edges. Semantic nodes comprise 21 files, 7 classes, 177 functions, 12 methods, 12 tests, 211 imports, and 5 decorators. Its edges include 347 calls, 371 references, 33 test-target links, 20 file imports, and one inheritance link. The implementation also emits route-handler, validator-model, and configuration-consumer links, but this corpus instantiates none of those framework-specific cases.

Each representation receives the same task text, lexical-anchor scorer, two-hop neighborhood, and cutoffs $k\in\{1,3,5,10\}$. We sweep packed-context budgets of 4K, 8K, 12K, and 24K characters. Semantic relevant-file recall is $.339,.339,.567,.828$ across that sweep, versus $.256,.256,.256,.289$ for both coarse baselines. Semantic relevant-symbol recall is $.478,.478,.883,.958$, versus $.556,.639,.778,.778$. Semantic packing uses definition-level chunks while the baselines pack whole files, so this comparison evaluates the complete representations, not graph edges in isolation. Our next experiment must cross file/symbol chunking with coarse/semantic edges factorially, on independently authored issue tasks, before we attribute the gain to edge type.

Tree-sitter makes syntax and source locations available but does not provide Python type resolution. Our resolver links same-file definitions, imported bindings, \texttt{self}/\texttt{cls} methods, and unique repository symbols conservatively. It omits ambiguous targets and can miss dynamic imports, reflection, dependency injection, monkey-patching, and runtime dispatch. These limitations are why we treat the semantic graph $\hat G_T^{\mathrm{sem}}$ as another proxy rather than as $G_T$ itself.

\section{Measurement and Benchmark Details}
\label{app:measurement}

We collect here the supporting measurements behind the trajectory results: the AUC bootstrap procedure, alternative engagement baselines, import-graph extraction, read volume in the bloat intervention, and SWE-bench execution.

\paragraph{AUC cells and bootstrap procedure.}
The residency analysis groups trials by \texttt{(agent, workload)}. A cell enters ROC AUC only if it holds at least one success and one failure, and we exclude one-class cells rather than assign them AUC $0.5$. For a per-cell interval, the implementation resamples the success and failure strata separately, preserving both classes in every one of 1,000 draws. Ordinary unstratified bootstrap samples with one class therefore never enter the interval.

\paragraph{Residency and engagement baselines.}
Table~\ref{tab:baselines} compares $\rho_{32}$ with coverage, volume, and effort baselines on an $n{=}12$ lean-versus-code-bloat block. We sign every entry in one direction, namely the AUC at which the lean condition scores higher, so that the rows stay mutually interpretable. In this block residency reaches AUC $0.83$, volume baselines remain near chance, and the two effort measures separate the conditions in the opposite direction. The twenty-trial rerun reported in the main paper reverses this pattern, so we keep the table as a record of the original block rather than as evidence that residency separates the conditions. This block is a dedicated run of six lean and six code-bloat trials, held apart from the intervention cells of Table~\ref{tab:bloat} of the main paper. No trial in it succeeds, and all twelve nonetheless pass 97.5--98.7\% of the suite, so success is invariant across the block and outcome cannot drive any of these comparisons. We recompute all values from the released ledger.

\begin{table}[htbp]
\centering
\footnotesize
\begin{tabular}{@{}lc@{}}
\toprule
Baseline & AUC (lean $>$ bloat, $n{=}12$) \\
\midrule
Residency $\rho_{32}$ (coverage)           & \textbf{0.83} \\
Distinct files edited (coverage)           & 0.82 \\
Number of reads (volume)                   & 0.58 \\
Number of tool calls (volume)              & 0.53 \\
Number of edits (volume)                   & 0.49 \\
Wall time (effort)                         & 0.33 \\
Self-rereads (effort)                      & 0.24 \\
\bottomrule
\end{tabular}
\caption{Condition-discrimination AUC on the $n{=}12$ lean-vs-bloat block, all signed as lean scoring higher. In this block the coverage measures separate the conditions and volume counts sit near chance, while the effort measures fall below $0.5$ because in-file bloat makes the agent reread more and work longer while covering less of the dependency neighborhood. The twenty-trial rerun in the main paper reverses this pattern. A self-reread is a read of a file the same agent had already read.}
\label{tab:baselines}
\end{table}

\paragraph{Import-graph extraction.}
We build the proxy graph $\hat G_T$ by abstract syntax tree (AST) traversal of every Python source, resolving each \texttt{import} and \texttt{from $\ldots$ import} via package-relative path. The extractor captures docstring-embedded imports, conditional imports, and re-exports through \texttt{\_\_init\_\_.py}, with a regex fallback for malformed workspaces. The cross-family repositories (\texttt{marked}, \texttt{ripgrep}) use language-specific extractors of the same shape. We canonicalize file paths to workspace-relative form before set comparison. For each edit on file $f$, $\rho_w$ is the fraction of the in- and out-neighbors of $f$ in $\hat G_T$ among the distinct files in the agent's tool-result stream within the last $w$ tool events.

\paragraph{Read volume per bloat condition.}
Table~\ref{tab:bloat-volume} reports median bytes read for the same conditions as Table~\ref{tab:bloat} of the main paper, which separates placement effects from raw reading volume.

\begin{table}[htbp]
\centering
\footnotesize
\begin{tabular}{@{}lrr@{}}
\toprule
Condition & Claude bytes & Codex bytes \\
\midrule
lean                        & 35{,}497  & 104{,}551 \\
document bloat              & 21{,}214  & 50{,}992 \\
code bloat                  & 24{,}200  & 152{,}270 \\
code bloat, single agent    & 11{,}774  & 501{,}605 \\
renamed twin                & 8{,}843   & 151{,}173 \\
\bottomrule
\end{tabular}
\caption{Median byte count of file content each agent read in the bloat-placement conditions. Pass rates appear in Table~\ref{tab:bloat} of the main paper.}
\label{tab:bloat-volume}
\end{table}

\paragraph{SWE-bench execution.}
Every trial ran under its harness and wrote harness-native transcripts when the harness exposed them. The Codex event extractor consumes the current \texttt{codex-cli 0.14x} session layout. We score Aider runs but exclude them from trajectory statistics because they expose no recoverable multi-edit stream in this setup.

\section{Run-to-Run Stability}
\label{app:stability}

We executed thirty (benchmark, agent, condition) cells twice in separate
batches on different dates, covering 404 trials. Full-marks
outcomes disagree across batches on 4.8\% of matched trials, which bounds
the noise under the contrasts we report.

\section{Execution Environment and Reproducibility}
\label{app:reproducibility}

We record here the execution details needed to reproduce the closed-book isolation: the sandbox's exact scope, the local-model serving environment, and the Codex run we excluded from the sandboxed front-load headline.

\paragraph{Recomputing the reported values.}
We recompute the following from the released per-trial ledgers and compare each against its printed value: the condition-discrimination baselines, the independent-fixes decomposition control and its wall-clock speedup, the estimator-enumeration summary, and the semantic-graph pilot including its graph statistics and budget sweep.

\paragraph{Project repository.}
The project repository is available at \url{https://github.com/mpi-dsg/agent-coherence}.

\paragraph{Closed-book sandbox.}
Every closed-book trial in the main paper's channel-control experiments runs the model process under an operating-system sandbox policy that we generate per trial and retain with the trial record. The policy denies all reads beneath the project root, which holds the workload generators, the reference outputs, and every prior trial transcript, and then re-admits two things the process needs to start: the trial's own scratch directory, and the command-line interface (CLI) configuration paths the tool probes at launch. We build it as a denylist over a permissive base, because the CLI must still load its binary, its libraries, and its interpreter.

Paths elsewhere on the machine stay readable, so we do not claim total isolation. The claim is narrower and sufficient: the answer specific to each workload exists only beneath the denied root. We hand-authored the four novel workloads, and their application programming interface (API) rules appear in no other file on the machine and in no training corpus.

\paragraph{Local-model cluster.}
We served DeepSeek-Coder-V2-Lite ($16$B parameters, $30$\,GB weights) and Qwen3-Coder-30B-A3B-Instruct ($57$\,GB weights) locally via vLLM $0.24$ on an H100 PCI Express (PCIe) graphics processing unit (GPU) node (2$\times$80\,GB, PCIe Gen5, 1\,TB main memory). DeepSeek uses one GPU at $85$\% memory use and its default sequence length. Qwen also uses one GPU, and we trim its default $262{,}144$-token context to $65{,}536$ tokens so the key-value (KV) cache fits alongside the weights. Neither model needs custom kernels.

\paragraph{Excluded Codex front-load diagnostic.}
The main front-loaded numbers in Table~\ref{tab:frontload} of the main paper exclude Codex: under the final sandbox policy its binary cannot complete a read it makes at startup. With the sandbox relaxed, Codex reaches $12/12$ in 12 of 24 trials, three of six per workload, so it recovers on every workload without matching the ceiling rate of the other families. The denied read happens before the model API call, so it depends on neither the prompt content nor the recovery behavior. We report this only as an instrumentation diagnostic. It does not enter the headline front-load evidence or support a capacity claim.

\section{Ethics and Data Statement}
\label{app:ethics}

We accessed all models through published inference APIs (Anthropic, OpenAI, Z.ai, Google) or as self-hosted open-weight checkpoints (DeepSeek-Coder-V2-Lite, Qwen3-Coder-30B) under their respective licenses. We involve no human subjects, no personally identifying data, and no confidential source code. We hand-authored the synthetic workloads specifically for this study, and they contain no code drawn from any repository.

Full per-trial data accompanies the code release: prompts, raw model responses, evaluation logs, and sandbox profiles. The closed-book matrix, front-load recovery and its reference comparison, test-identity Jaccard, SWE-bench resolved rates, matched event simulator, and semantic-graph pilot are recomputable from that released data. The residency baselines, the fault-injection sweep, the conflicting-source cells, and the SWE-bench trajectory statistics we recompute separately, from their own analysis code rather than from the shared recomputation script.

The tool-using intervention cells, namely the residency baselines, the overlay sign flip, the bloat matrix, and the capacity and cross-harness counts, ship as derived per-trial records. The main paper marks the single-intervention cells among them as exploratory rather than confirmatory. The Execution Environment and Reproducibility section above records the sandbox profile, local-model cluster, and excluded Codex diagnostic needed to reproduce the execution environment.

\end{document}